\documentclass[10pt,twocolumn,letterpaper]{article}

\usepackage{iccv}
\usepackage{times}
\usepackage{epsfig}
\usepackage{graphicx}
\usepackage{amsmath}
\usepackage{amssymb}

\usepackage{booktabs}
\usepackage{multirow}
\usepackage{adjustbox}
\usepackage{amsthm}
\usepackage{mathtools}
\usepackage{commath}
\usepackage{array}
\usepackage[caption=false]{subfig}

\DeclareMathOperator{\E}{\mathbb{E}}

\usepackage[pagebackref=true,breaklinks=true,letterpaper=true,colorlinks,bookmarks=false]{hyperref}

\usepackage[capitalize]{cleveref}
\crefname{section}{Sec.}{Secs.}
\Crefname{section}{Section}{Sections}
\Crefname{table}{Table}{Tables}
\crefname{table}{Tab.}{Tabs.}

\iccvfinalcopy % *** Uncomment this line for the final submission

\def\iccvPaperID{9887} % *** Enter the ICCV Paper ID here

\ificcvfinal\fi

\begin{document}

%%%%%%%%% TITLE
\title{Self-supervised Image Denoising with Downsampled Invariance Loss and Conditional Blind-Spot Network}

\author{Yeong Il Jang$^1$~~~~~~~~
	    Keuntek Lee$^1$~~~~~~~~
	    Gu Yong Park$^1$~~~~~~~~
	    Seyun Kim$^2$~~~~~~~~
	    Nam Ik Cho$^1$\\
$^1$Department of ECE, INMC, Seoul National University ~~~~~ $^2$Gauss Labs Inc.\\
{\tt\small \{jyicu, leekt000, pgy9134\}@snu.ac.kr, seyun.kim@gausslabs.ai, nicho@snu.ac.kr}}

\maketitle
\ificcvfinal\thispagestyle{empty}\fi

%%%%%%%%% ABSTRACT
\begin{abstract}
	\label{abstract}
	There have been many image denoisers using deep neural networks, which outperform conventional model-based methods by large margins. Recently, self-supervised methods have attracted attention because constructing a large real noise dataset for supervised training is an enormous burden. The most representative self-supervised denoisers are based on blind-spot networks, which exclude the receptive field's center pixel. However, excluding any input pixel is abandoning some information, especially when the input pixel at the corresponding output position is excluded. In addition, a standard blind-spot network fails to reduce real camera noise due to the pixel-wise correlation of noise, though it successfully removes independently distributed synthetic noise. Hence, to realize a more practical denoiser, we propose a novel self-supervised training framework that can remove real noise. For this, we derive the theoretic upper bound of a supervised loss where the network is guided by the downsampled blinded output. Also, we design a conditional blind-spot network (C-BSN), which selectively controls the blindness of the network to use the center pixel information. Furthermore, we exploit a random subsampler to decorrelate noise spatially, making the C-BSN free of visual artifacts that were often seen in downsample-based methods. Extensive experiments show that the proposed C-BSN achieves state-of-the-art performance on real-world datasets as a self-supervised denoiser and shows qualitatively pleasing results without any post-processing or refinement.
\end{abstract}

\begin{figure}[t]
	\centering
	\captionsetup[subfigure]{justification=centering}
	\subfloat[Noisy\\28.48dB / 0.9011~\label{fig:intro:noisy}]	{\includegraphics[width=.48\linewidth]{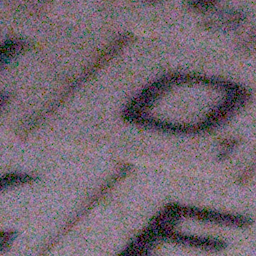}}
	\hfill
	\subfloat[CVF-SID~\cite{neshatavar2022cvf} \\ 34.21dB / 0.9381\label{fig:intro:cvf_sid}]	{\includegraphics[width=.48\linewidth]{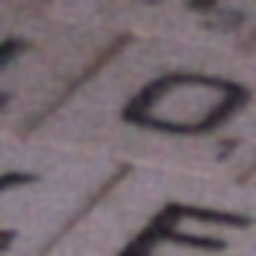}}
	\\
	\subfloat[AP-BSN~\cite{lee2022ap} \\34.45dB / 0.9081\label{fig:intro:ap_bsn}]	{\includegraphics[width=.48\linewidth]{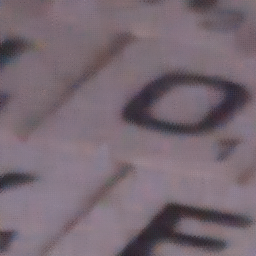}}
	\hfill
	\subfloat[C-BSN (Ours)\\\textbf{36.31dB / 0.9483}\label{fig:intro:ours}]
	{\includegraphics[width=.48\linewidth]{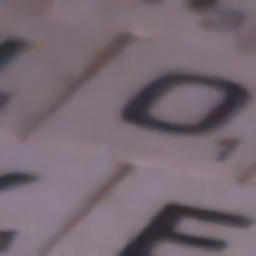}}
	\vspace{2mm}
	\caption{\textbf{Visual comparison of denoised images on SIDD validation~\cite{abdelhamed2018high}.} 
		Our C-BSN shows better details and no artifacts without post-processing or refinement.
		Best viewed in pdf.
	}
	\vspace{-5mm}
	\label{fig:intro}
\end{figure}

%%%%%%%%% BODY TEXT
\section{Introduction}
\label{sec:intro}
Image denoising aims to recover a clean image from its corrupted counterpart. Recently, image denoisers using convolutional neural networks (CNNs) have achieved great performances, significantly outperforming conventional model-based ones \cite{zhang2017beyond,zhang2018ffdnet,tai2017memnet}. They trained networks by minimizing the difference between the network outputs and the ground-truth clean images. In early works, they assumed the camera noise as an additive white Gaussian noise (AWGN) and generated a large number of clean-noisy image pairs for the supervised training. However, the denoisers trained with AWGN fail to generalize to real-world camera noises due to the difference between the Gaussian and real noise distributions~\cite{guo2019toward}.
Specifically, real noise follows a more complicated distribution than a simple Gaussian and gets more correlated spatially and chromatically while passing through an in-camera image processing pipeline, such as demosaicing that involves the computation using adjacent pixels.

Some researchers attempted to find a more realistic noise model to deal with real noise. In the case of camera-raw images, noise can be modeled with a relatively simple distribution such as heteroscedastic Gaussian~\cite{foi2008practical}. Hence, a raw image added with such synthetic noise is passed through a camera image signal processor (ISP) model to generate a realistic noisy sRGB image \cite{guo2019toward,zamir2020cycleisp}. Other works synthesized realistic noise using generative models \cite{chen2018image,cai2021learning,jang2021c2n,abdelhamed2019noise}.
Another approach is to construct paired real noise datasets from real photos like DND~\cite{plotz2017benchmarking} and SIDD~\cite{abdelhamed2018high}. Training in a supervised manner with those datasets successfully reduced the noise of real cameras \cite{anwar2019real,zamir2020learning,zamir2021multi}. However, acquiring aligned clean images corresponding to noisy ones requires a series of static photos of the same scene. It is costly or even impossible in some cases, such as medical images, since it requires strictly controlled capturing and complicated post-processing. Also, since they used several cameras in specific environments for capturing real noises, they might have different distributions from the ones captured from other cameras and from the same cameras with different shooting environments.

To mitigate the necessity of large aligned datasets, self-supervised denoising that requires only noisy images has been proposed. The most representative methods are based on blind-spot networks (BSN), where each output pixel is estimated from the surrounding noisy pixels except for the corresponding one. It enables the network to learn with the self-supervised loss function, where the same noisy images are used as both input and target. The idea of blind-spot prevents the network from converging to a trivial identity function. The BSN is shown to converge to the clean image under the assumption that the expectation of the noise is zero and the noise is pixel-wise independent. They imposed blindness to the network by masking the input image \cite{krull2019noise2void,batson2019noise2self} or by designing networks that structurally exclude the central pixel from the receptive fields \cite{laine2019high,wu2020unpaired,lee2022ap}. However, the BSN-based self-supervised algorithms have two limitations; 1) The network cannot utilize the center pixel which is the most informative. 2) It is not applicable to real noise since it has a pixel-wise correlation in the sRGB domain~\cite{lee2022ap}. 

In this paper, we propose a novel self-supervised learning framework to denoise real noise without the blind-spot, \ie, with the center pixel information. Our framework overcomes the above-stated limitations by deriving a novel downsampled invariance loss function. The downsampled invariance loss employs a novel conditional blind-spot network (C-BSN) and random subsampler. Specifically, our C-BSN conditionally controls its blindness by switching the masked convolution operations. It allows the network to be regularized by its blind-spot counterpart, which prevents the trivial solution. Furthermore, we impose the loss on randomly downsampled subimage so that the correlation of the noise is weakened without inducing visual artifacts. In addition, we augment the loss with a blind self-supervised loss for stabilizing the training. Extensive experiments have been conducted to evaluate the proposed framework, which validates that the C-BSN outperforms existing self-supervised denoisers and even some supervised methods trained with real noise datasets.

The contributions of our method are summarized as follows:
\begin{itemize}
	\item We propose a novel self-supervised denoising framework that can be processed without a blind-spot. We theoretically derive the upper bound of the self-supervised loss as downsampled invariance loss, which exploits masked output as the regularization of the denoised image without masking. In addition, the proposed method does not require post-processing or noise statistics.
	\item To apply downsampled invariance loss, we propose a novel conditional blind-spot network named C-BSN, which conditionally controls the blindness of the network. To deal with the spatial correlation of the real camera noise, a random subsampler is proposed to avoid visual artifacts.
	\item The C-BSN shows state-of-the-art performance in real-world sRGB benchmarks DND~\cite{plotz2017benchmarking} and SIDD~\cite{abdelhamed2018high}, as shown in Figs.~\ref{fig:intro}, \ref{fig:qual}, and \ref{fig:qual_sidd}.
\end{itemize}

\section{Related Works}
\label{sec:related}
\noindent\textbf{Deep Image Denoising} \label{subsec:deepdenoising}
Image denoisers based on Convolutional Neural Networks (CNNs) have outperformed conventional model-based algorithms. In early works, deep image denoisers were trained with large datasets consisting of clean images and noisy ones corrupted by synthetic Gaussian noise. DnCNN~\cite{zhang2017beyond} proposed a CNN denoiser with batch normalization and residual learning. Following DnCNN, many networks with more sophisticated architectures have been proposed \cite{zhang2018ffdnet,mao2016image,tai2017memnet,liu2018non,zhang2019residual}. However, denoisers trained with synthetic Gaussian noise could not generalize well for denoising real-world noisy images. To alleviate this problem, CBDNet~\cite{guo2019toward} synthesized heteroscedastic Gaussian noise and processed it through the camera ISP model. Some works simulated realistic noise using generative adversarial network (GAN)~\cite{chen2018image,cai2021learning,jang2021c2n} or flow-based methods \cite{abdelhamed2019noise,maleky2022noise2noiseflow,kousha2022modeling}.
With the development of real-world sRGB datasets~\cite{abdelhamed2018high,plotz2017benchmarking}, recent denoisers have been trained and tested on these datasets, \cite{anwar2019real,zamir2020learning,zamir2021multi,yue2019variational,
	jang2020dual,kim2020transfer,soh2022variational,tu2022maxim}, demonstrating that the real noisy images could be successfully denoised. Moreover, it has been shown that earlier denoisers can also work better by retraining with these datasets. However, collecting a large dataset is laborious and costly. Moreover, the networks trained with a specific dataset may not function properly on images captured by other cameras, not included in the dataset, or images from other domains, such as medical, electron, and ultra-sonic.

\noindent\textbf{Self-supervised Deep Image Denoising} 
In order to overcome the lack of aligned real noisy-clean image pairs, self-supervised learning that trains denoiser with solely noisy images has been proposed. Lehtinen \etal~\cite{lehtinen2018noise2noise} proposed Noise2Noise where training pairs are two noisy images of the same scene. 
Noise2Void~\cite{krull2019noise2void} and Noise2Self~\cite{batson2019noise2self} introduced self-supervised denoisers that require only single noisy images by masking the center pixel of the receptive field.
Without masking input pixels, Laine \etal~\cite{laine2019high} proposed a structurally blind-spotted network with a concatenation of half-plane receptive field U-Nets~\cite{ronneberger2015u}. 
Wu \etal~\cite{wu2020unpaired} introduced dilated blind-spot network (D-BSN), where masked convolution is followed by dilated convolutions and $1\times1$ convolutions, strictly excluding the center pixel from the receptive field. Self2Self~\cite{quan2020self2self} trained the denoiser with a single noisy image by applying Bernoulli dropout.
Neighbor2Neighbor~\cite{huang2021neighbor2neighbor} proposed a self-supervised loss between two subsampled images. Also, assuming known noise characteristics, Noisy-as-clean~\cite{xu2020noisy} and Noisier2noise~\cite{moran2020noisier2noise} added a proper noise to the noisy image and used the pair as a training set. Recorrupted2Recorrupted~\cite{pang2021recorrupted} generated pairs of Gaussian-corrupted images to be used as training pairs. In general, real noises of the sRGB domain have unknown or non-stationary statistics and are spatially correlated, making the above methods less applicable.

Recently, some works have been proposed to overcome the limitations of the above BSN-based methods. 
To mitigate the spatial correlation of real noise, AP-BSN~\cite{lee2022ap} utilized pixel downshuffle (PD)~\cite{zhou2020awgn} asymmetrically. 
During training, the network was trained using high strides where the assumption of independence holds. During testing, low strides were used to preserve more pixel information.
CVF-SID~\cite{neshatavar2022cvf} disentangled a clean image and signal-dependent noise from real-world noisy input.
To utilize information of center pixel, Laine \etal~\cite{laine2019high} post-processed the denoised output to be the posterior with the known noise model in a Bayesian approach.
Noise2Same~\cite{xie2020noise2same} derived the upper bound of self-supervised loss without introducing the blind-spot. Blind2Unblind~\cite{wang2022blind2unblind} proposed re-visible loss that makes blind-spot visible again. However, to the best of our knowledge, there has been no research that handles both problems (use of blind-spot and handling spatial correlation) for self-supervised image denoising.

\section{Method}
\label{sec:method}

\begin{figure*}[t]
	\centering
	\includegraphics[width=\linewidth]{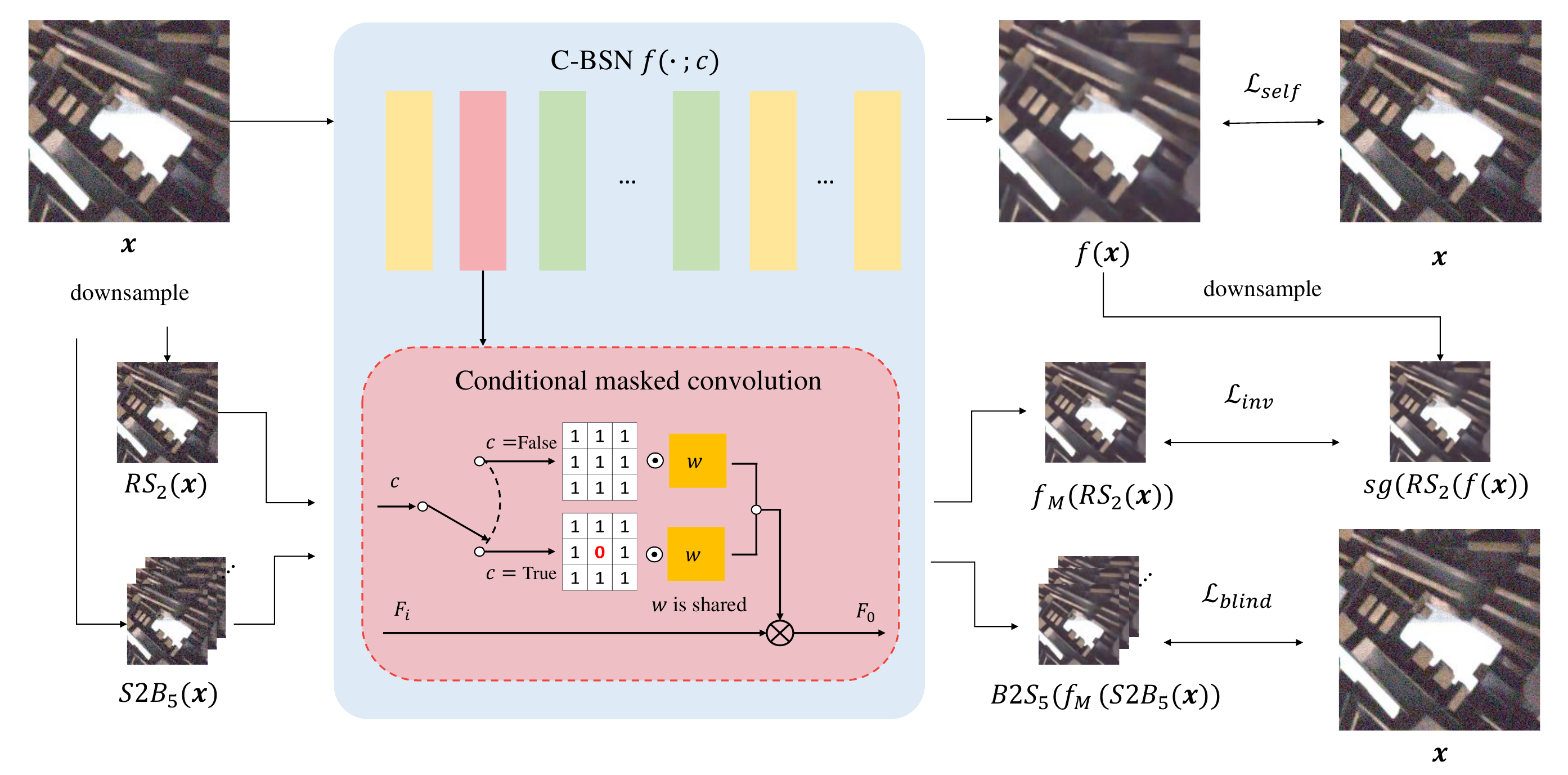}
	\caption{\textbf{Overview of the proposed C-BSN framework.}
		Illustration of the C-BSN architecture and loss functions.
		For simplicity, condition variable $c$ is omitted in $f$ when $c=$ False, and $f_M$ denotes the blind-spot network with $c=$ True. 
		The yellow box represents $1\times 1$ convolution and the green box represents dilated convolution module, which consists of dilated convolution followed by $1\times 1$ convolution and residual skip connection.
		Note that RS samples the same pixel indices when calculating downsampled invariance loss with $RS_2(f(\boldsymbol{x}))$ and $sg(f_M(RS_2(\boldsymbol{x}))$.}
	\label{fig:overview}
\end{figure*}
\subsection{Overview}
\label{subsec:Overview}
We introduce a novel self-supervised learning framework to denoise real-world RGB images, which is illustrated in \cref{fig:overview}.
We propose a novel loss function that can be directly optimized on the input image without loss of information. It consists of self-supervised loss and downsampled invariance loss that controls the extent of the blindness. Our main idea of the downsampled invariance loss is to make a blind-spot network serve as regularization of the same network while preserving network parameters. To this end, we propose a conditional blind-spot network (denoted C-BSN in the figure) to selectively mask the center pixel in the receptive field. In addition, we introduce Random Subsampler (RS) to decorrelate noise spatially. The pixel-shuffle downsampling (PD)~\cite{zhou2020awgn} also loosens the spatial correlation of the noise, but it generates severe checkerboard artifacts. On the contrary, since our RS draws a pixel randomly from each grid, it does not produce such artifacts. We denote the noisy input image as $\boldsymbol{x}$ and the corresponding clean image as $\boldsymbol{y}$. For brevity, the channel dimension is omitted, and spatial dimensions are vectorized, \ie, $\boldsymbol{x},\boldsymbol{y}\in\mathbb{R}^m$.

\subsection{Revisiting Noise2Same}
Under the assumption that noise is zero mean and pixel-wise independent, Baston \etal~\cite{batson2019noise2self} proved that self-supervised loss is equivalent to supervised loss if the network is $\mathcal{J}$-invariant.

\newtheorem{definition}{Definition}
\begin{definition}\cite{batson2019noise2self}
	Let $\mathcal{J}$ be a partition of the dimensions $\{1,...,m\}$ and let $J \subset \mathcal{J}$. A function $f: \mathbb{R}^m \rightarrow \mathbb{R}^m$ is $J$-invariant if $f(\boldsymbol{x})_J$ does not depend on the value of $\boldsymbol{x}_J$. It is $\mathcal{J}$-invariant if it is $J$-invariant for each $J \in \mathcal{J}$.
\end{definition}
\noindent Subscripted notation $\boldsymbol{x}_J$ is used for $\boldsymbol{x}$ restricted to $J$.
Noise2Same~\cite{xie2020noise2same} analyzed that strictly $\mathcal{J}$-invariant function is not optimal for the denoisers. Rather, it mitigates the $\mathcal{J}$-invariance constraints by minimizing the upper bound of supervised loss,
\begin{equation}
	\begin{aligned}
		\mathcal{L}_{N2Same} &= \E_{\boldsymbol{x}}||f(\boldsymbol{x})-\boldsymbol{x}||^2/m \\
		&\hspace{-2mm}+\lambda_{inv} \E_J(\E_{\boldsymbol{x}} ||f(\boldsymbol{x})_J -f(\boldsymbol{x}_{J^C})_J||^2/|J|)^{\frac{1}{2}},
		\label{eq:self}
	\end{aligned}
\end{equation}
where $\boldsymbol{x}$ is the normalized input image so that the mean of $\boldsymbol{x}$ is zero and the standard deviation equals one. The first term is the self-supervised loss, while the second term controls how $\mathcal{J}$-invariant $f$ should be.

\subsection{Downsampled Invariance Loss}
\label{subsec:invloss}
Noise2Same upper bound holds when $f(\boldsymbol{x}_{J^C})$ in \cref{eq:self} is not correlated with $\boldsymbol{x}_J$.
Although the pixel-wise independent noise such as AWGN satisfies the above constraint, real noise is correlated spatially, which makes it no longer applicable. Instead of randomly sampling the subset $J$, we sample the downsampled image to reduce the correlation, following previous research \cite{zhou2020awgn,lee2022ap}. Precisely, we propose modified version of ~\cref{eq:self} as follows:

\newtheorem{proposition}{Proposition}
\begin{proposition}\label{proposition}
	Let $\boldsymbol{x}$ be a normalized zero-mean noisy image conditioned on $\boldsymbol{y}$, $\E[\boldsymbol{x}|\boldsymbol{y}]=\boldsymbol{y}$. 
	Let $d$ be any downsampling operation and $d_s(\boldsymbol{x})$ be a set of downsampled pixels of $\boldsymbol{x}$ with a stride of $s$.
	Assume that downsampled subimage $d_s(\boldsymbol{x})$ has zero pixel-wise correlation and $f_M$ is a blind-spot network. Then, the following inequality holds.
	\begin{align}
		&\E_{\boldsymbol{x},\boldsymbol{y}} \norm{f(\boldsymbol{x})-\boldsymbol{y}}^2 + \norm{\boldsymbol{x}-\boldsymbol{y}}^2 \leq 
		\E_{\boldsymbol{x}} \norm{f(\boldsymbol{x})-\boldsymbol{x}}^2\nonumber &&\\
		&\hspace{6mm}+ 2\sqrt{ms^2} \mathop{\E}_{d_s(\boldsymbol{x})}[\E \,\norm{d_s(f(\boldsymbol{x})) - f_{M}(d_s(\boldsymbol{x}))}^2]^{\frac{1}{2}}.
		\label{eq:prop}
	\end{align}
\end{proposition}
\noindent
\Cref{proposition} provides the upper bound of the supervised loss with the self-supervised loss and the regularization of the downsampled output with the blind output of the downsampled input. 
We prove in the supplementary material that $f(\boldsymbol{x}_{J^C}) $ in \cref{eq:self} can be replaced by $f_M(d_s(\boldsymbol{x}))$, which has no correlation with $d_s(\boldsymbol{x})$.
This simplifies the second term of \cref{eq:self} to our new downsampled invariance loss,
\begin{equation}
	\mathcal{L}_{inv} = \sqrt{\frac{s^2}{m}} \norm{d_s(f(\boldsymbol{x})) -sg(f_M(d_s(\boldsymbol{x})))}_2,
	\label{eq:self2}
\end{equation}
where $sg$ is a stop-gradient operation.
With \Cref{proposition}, we can optimize the denoising network by minimizing the right side of \cref{eq:prop}.
Details of the proof are in the supplementary material. 

\subsection{Conditional Blind-Spot Network}
\label{subsec:cbsn}
\Cref{eq:self2} requires the parameters of the network $f$ to be shared regardless of the blind-spot. In the case of Noise2Same~\cite{xie2020noise2same}, the network remains unchanged as blindness is caused by masking input pixels, not by the network structure. However, masking causes train-test discrepancy of inputs and harms training efficiency since loss can be back-propagated only through masked pixels. On the other hand, a network such as D-BSN~\cite{wu2020unpaired} excludes the center pixel by its architecture. It can be optimized through every single pixel, though the blindness cannot be removed. To control blindness with D-BSN architecture conditionally, the network structure should be changed while sharing the training parameters. To this end, we propose a conditional blind-spot network (C-BSN) to make a blind-spot without masking the input image.

In D-BSN, blindness is induced by masked convolutions, and dilated convolutions prevent masked pixel information from being mixed in. We switch the behavior of masked convolution by changing the mask of kernels according to the given condition $c$:
\begin{equation}
	F_o = (M \odot W)*F_i+b,
	\label{eq:condconv}
\end{equation}
\begin{equation}
	M = 
	\begin{cases}
		\mathbf{1}_{k\times k} - \boldsymbol{\delta}_{k\times k},  & \text{if } c=\text{True},\\
		\mathbf{1}_{k\times k},  & \text{otherwise,}
	\end{cases}
\end{equation}
where $W$ is convolutional filter, $b$ is a bias, and $F_i$ and $F_o$ are input features and output features, respectively.
$\boldsymbol{\delta}_{k\times k}$ is a $k \times k$ Dirac delta kernel, and $\mathbf{1}$ is the matrix of ones.
For simplicity, we omit the condition variable $c$ when $c=$ False and represent only blind-conditioned network as $f_M = f(\cdot;c=\mathrm{True})$.
We only use $f_M$ in the training phase, and all test images are inferred by non-blind network $f$ without the loss of information.

Applying conditional masked convolution can alter the output features' distribution because the kernel's center is set to zero when $c$ is False. However, $f$ should be trained differently from $f_M$ to utilize the masked pixel. In addition, the center of the kernel is trained independently of $f_M$, based on the modified feature distribution. Hence, we use the kernel and its mask without normalization between $c=$True and $c=$False.

\subsection{Random Subsampler}
\begin{figure}[]
	\centering
	\includegraphics[width=.95\linewidth]{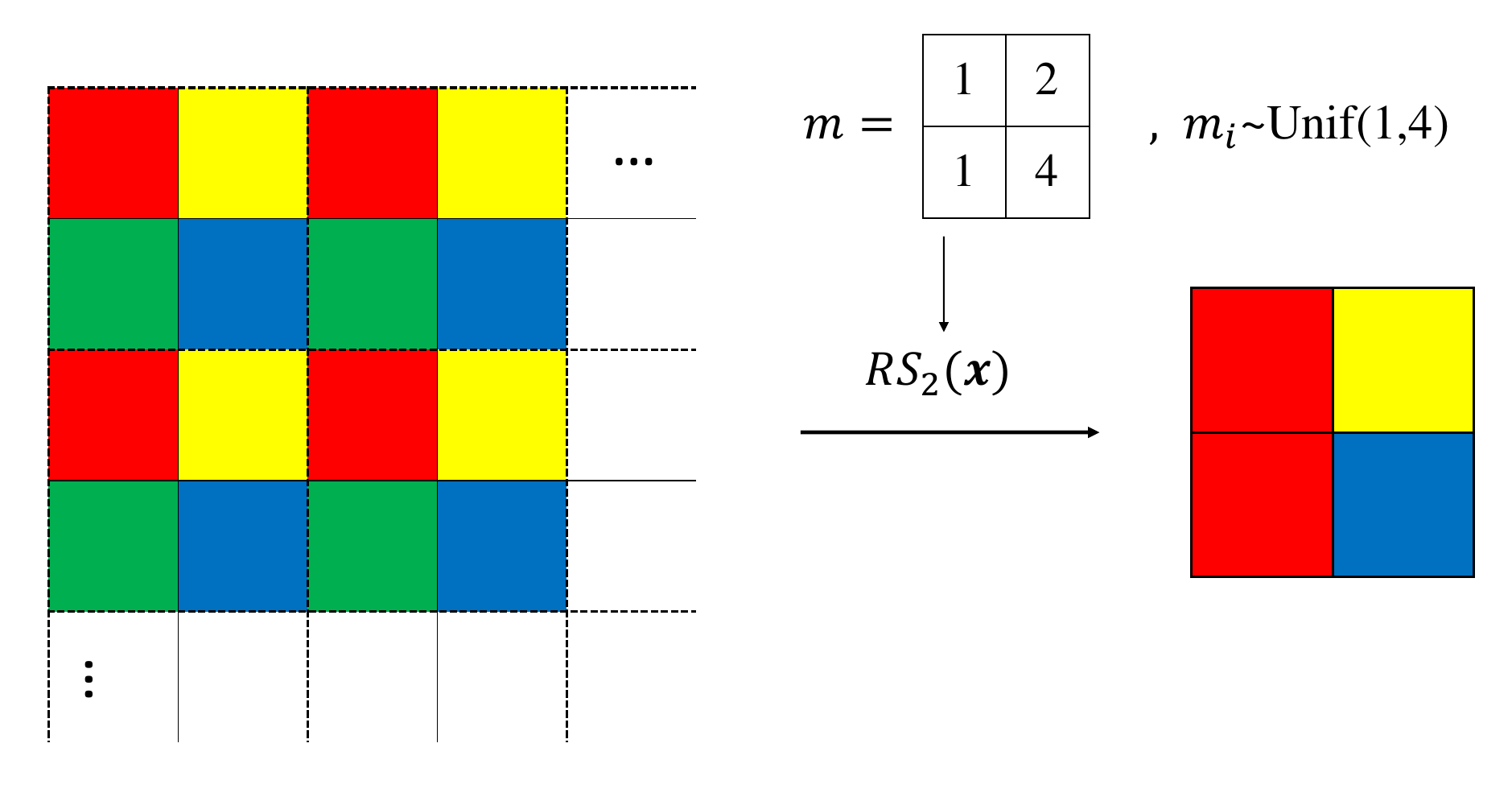}
	
	\caption{\textbf{Details of a random subsampler with a stride of two.}
		Each color represents a relative position within the cell.
		$m$ is a selection mask introduced for an explanation.
		The indices of $m$ that determine which pixel will be selected are randomly sampled from the uniform distribution.
	}
	\label{fig:downshuffle}
\end{figure}

\label{subsec:downsampling}
In \Cref{subsec:invloss}, we introduced a downsample operator to the invariance loss to extract a subset of the image with zero spatial correlation. This constraint is guaranteed in pixel-wise synthetic noises, while real noise does not comply.
In order to remove spatial pixel dependency, Zhou \etal~\cite{zhou2020awgn} and Lee \etal~\cite{lee2022ap} utilized pixel-shuffle downsampling (PD).
The PD is the inverse operation of the pixel-shuffle~\cite{shi2016real} and creates the mosaic of the subimages. However, directly applying PD in downsampled invariance loss is not trivial since the expectation of \cref{eq:prop} is calculated over subimages $d_s(\boldsymbol{x})$. 
Another approach to decorrelate the noise is a space2batch (S2B) operation, where pixel down-shuffled subimages are concatenated along batch dimension instead of channel dimension.
However, naively applying S2B induces severe visual artifacts in the results. When S2B images are taken as input, 
all the subimages are calculated independently, which results in a checkerboard pattern in the batch2space (B2S) upsampled outputs, giving false guidance to the $f(x)$.

To deal with this problem, we propose a random subsampler $RS_s(\cdot)$, a subsampling operator to avoid the checkerboard artifact.
\Cref{fig:downshuffle} shows the details of our random subsampler.
Taking stride of two as an example, input images are divided into $2\times 2$ grid cells. For each cell, a pixel is randomly drawn within the cell, making $s$ times downsampled image.
If the randomly downsampled pixel in the adjacent cell is also adjacent, the correlation may occur significantly. However, in this case, the average distance from the other peripheral pixels becomes large, and the expected average distance between subsampled pixels can still be approximated to $s$. Therefore, as with PD, the expected spatial correlation is weakened by the random subsampler.

\subsection{Total Loss function}\label{subsec:totalloss}
In this section, we provide the total loss function.
For simple notation, we use $\norm{\cdot}$ to represent the pixel-averaged $L_1$ norm.
We substitute mean squared errors to the $L_1$ norm in the self-supervised loss, as
\begin{equation}
	\mathcal{L}_{self} = \norm{f(x) - x}.
	\label{eq:loss_self}
\end{equation}
Also, we find it beneficial to replace the root mean square (RMS) of the downsampled invariance loss with the $L_1$ norm as well and to use a random subsampler as a downsampling operation,
\begin{equation}
	\mathcal{L}_{invRS} = \norm{RS_2(f(\boldsymbol{x}))-sg(f_M(RS_2(\boldsymbol{x}))}.
	\label{eq:loss_di}
\end{equation}
From the \Cref{proposition} in \Cref{subsec:invloss}, we minimize the upper bound of supervised loss function,
\begin{equation}
	\mathcal{L}_{CBSN} = \mathcal{L}_{self} + \lambda_{inv}\cdot\mathcal{L}_{invRS}
	\label{eq:loss_CBSN}
\end{equation}
where $\lambda_{inv}$ is a hyperparameter to control the contribution of the downsampled invariance loss. We set the stride of RS as 2 in order to reflect more spatial information.

In addition, we introduce a self-supervised loss of the blind conditioned network, $\mathcal{L}_{blind}$, to stabilize the training as in \cite{wang2022blind2unblind}, where
\begin{equation}
	\mathcal{L}_{blind} = \norm{B2S_5(f_M(S2B_5(\boldsymbol{x}))) - \boldsymbol{x}}.
	\label{eq:loss_blind}
\end{equation}
While downsampled invariance loss utilizes the stride of two, the stride in \cref{eq:loss_blind} is five since the ideal BSN should be trained with as little correlation as possible.
Without blind self-supervised loss, $f_M(\boldsymbol{x})$ is random in the early stage of training, giving wrong guidance to the $f(\boldsymbol{x})$. 
Thus, we augment $\mathcal{L}_{CBSN}$ with the blind self-supervised loss to facilitate the transition from $f_M$ to $f$.
Additionally, we adopt warm-up scheduling to $\mathcal{L}_{CBSN}$. Scheduling parameter $\lambda_{sch}$ is multiplied to $\mathcal{L}_{CBSN}$, gradually increasing the impact of $\mathcal{L}_{CBSN}$. With all these in consideration, the total objective function is defined as
\begin{equation}
	\mathcal{L}_{total} = \mathcal{L}_{blind} + \lambda_{sch}\cdot\mathcal{L}_{CBSN}.
	\label{eq:loss_total}
\end{equation}

%-------------------------------------------------------------------------

\section{Experimental results}
\label{sec:result}
\begin{table*}[]
	\caption{\textbf{Quantitative comparison on SIDD and DND benchmarks.}
		PSNR and SSIM are from the official SIDD and DND websites.
		We use $\dagger$ notation to indicate that the network is trained on the test set directly. $*$ denotes that the method uses a self-ensemble strategy.
		The highest PSNR and SSIM of self-supervised algorithms are highlighted in \textbf{bold}.
	}
	\centering	
	
	\begin{tabular}{>{\centering}p{0.15\textwidth}|>{\centering}p{0.25\textwidth}|>{\centering}p{0.10\textwidth}			>{\centering}p{0.10\textwidth}>{\centering}p{0.10\textwidth}>{\centering\arraybackslash}p{0.10\textwidth}}
		\toprule
		\multirow{2}{*}{Supervision}&\multirow{2}{*}{Method}  & \multicolumn{2}{c}{SIDD}  & \multicolumn{2}{c}{DND}  \\
		&& PSNR(dB)  & SSIM     & PSNR(dB)  & SSIM \\ \midrule
		\multirow{2}{*}{Model-based} &BM3D~\cite{dabov2007image}&      25.65    &  0.685        &    34.51       &     0.851     \\
		&WNNM~\cite{gu2014weighted}&      25.78    &  0.809        &    34.67      &     0.865     \\ \midrule
		\multirow{7}{*}{Supervised}
		&DNCNN~\cite{zhang2017beyond}&      35.13     &  0.896        &    37.89       &     0.932     \\
		&CBDNet~\cite{guo2019toward}&      33.28    &  0.868        &    38.05       &     0.942     \\
		&RIDNet~\cite{anwar2019real}&      38.70     &  0.950        &    39.24       &     0.952     \\
		&AINDNet (R)*~\cite{kim2020transfer}&      38.84     &  0.951        &    39.34       &     0.952     \\
		&VDN~\cite{yue2019variational}&      39.26    &  0.955        &    39.38       &     0.952     \\
		&MIRNet~\cite{zamir2020learning}&     39.72     &    0.959      &    39.88        &     0.956     \\
		&MAXIM-3S~\cite{tu2022maxim}&     39.96     &    0.960      &    39.84        &     0.957     \\ \midrule
		\multirow{2}{*}{Generation-based}	& GCBD~\cite{xu2020noisy} &      -     &      -    &      35.58    &   0.922       \\
		&C2N*~\cite{jang2021c2n} + DIDN ~\cite{yu2019deep}&      35.35     &  0.937        &    36.38       &     0.887     \\ \midrule
		\multirow{8}{*}{Self-supervised}	& NAC~\cite{xu2020noisy} &      -     &      -    &      36.20     &   0.925       \\
		&R2R~\cite{pang2021recorrupted}&      34.78     &  0.898        &    -       &     -     \\
		&CVF-SID(T)~\cite{neshatavar2022cvf}&      34.43     &  0.912   &   36.31 &  0.923 \\
		&CVF-SID($\mathrm{S}^2$)$^\dagger$~\cite{neshatavar2022cvf}&      34.71     &  0.917   &   36.50 &  0.924 \\
		&AP-BSN~\cite{lee2022ap}	&     34.90      &    0.900      &     37.46      &  0.924       \\
		&AP-BSN + $\mathrm{R}^3$~\cite{lee2022ap}&    35.97      &    0.925      &     38.09      &  0.937  \\ 
		&C-BSN	&   36.82      &    \textbf{0.934}      &     38.45      &  0.939       \\     
		&C-BSN$^\dagger$	&    \textbf{36.84}      &    0.933      &     \textbf{38.60}      &  \textbf{0.941}\\
		\bottomrule  
	\end{tabular}
	
	\label{table:Result}
\end{table*}

\subsection{Implementation Details}
\label{subsec:implement}
We train and test our method on real-world sRGB camera noise. Our model is trained in two settings; one is trained with an external dataset, and the other is trained with a test set directly. For the external training set, we use the SIDD medium set~\cite{abdelhamed2018high}, which contains 320 pairs of aligned real noisy-clean images captured by five smartphone cameras.
We only use the noisy images as training samples and discard all clean images.
In addition, as C-BSN requires only noisy images to be trained, we train C-BSN$^\dagger$ solely on test set images.
We test the proposed algorithm in DND~\cite{plotz2017benchmarking} and SIDD~\cite{abdelhamed2018high} benchmark. DND consists of 50 high-resolution noisy images from four different cameras. Note that both benchmarks evaluate PSNR and SSIM online and do not provide ground truth images.

We crop $240 \times 240$ patches from training images and use the mini-batch size of 4.
We randomly rotate $90^{\circ}$ and flip for data augmentation for each image patch.
Input images are normalized so that the mean and the standard deviation are 0 and 1, respectively. The standard deviation is calculated as $\max(\mathrm{std},\frac{1}{\sqrt{m}})$ to avoid division by zero.

We follow the AP-BSN structure~\cite{lee2022ap} with modified masked convolution in order to compare the effectiveness of loss functions only. We set $\lambda_{inv}$ to 2 as derived in \Cref{proposition} and employ a warm-up strategy for $\lambda_{sch}$ that linearly increases from 0 to 1 for the first 200,000 iterations. 
We use Adam~\cite{kingma2015adam} optimizer with the initial learning rate 1e-4.
C-BSN is optimized for 400,000 iterations, and the learning rate is halved every 100,000 iterations, capped at 2e-5. Note that C-BSN requires a single inference of input image, and the downsampling operation is not performed in test time.

\begin{figure*}[t!]
	\centering
	\captionsetup[subfigure]{justification=centering}
	\subfloat[Noisy image\\
			26.90 / 0.7527\label{fig:noisy_dnd}]
	{\includegraphics[width=.19\linewidth]{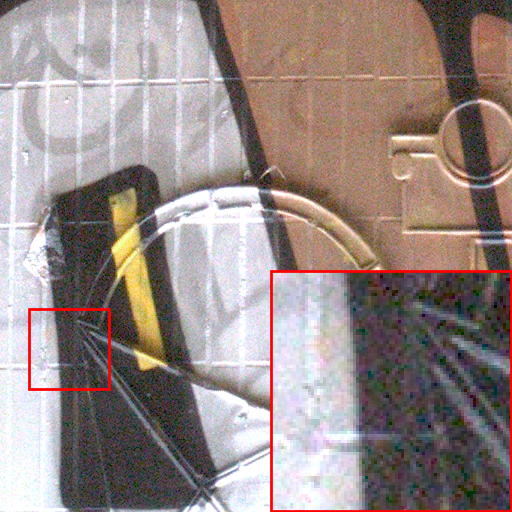}}
	\hfill
	\subfloat[CVF-SID(S$^2$)~\cite{neshatavar2022cvf}\\
		28.74/0.8737\label{fig:cvf_dnd}]
	{\includegraphics[width=.19\linewidth]{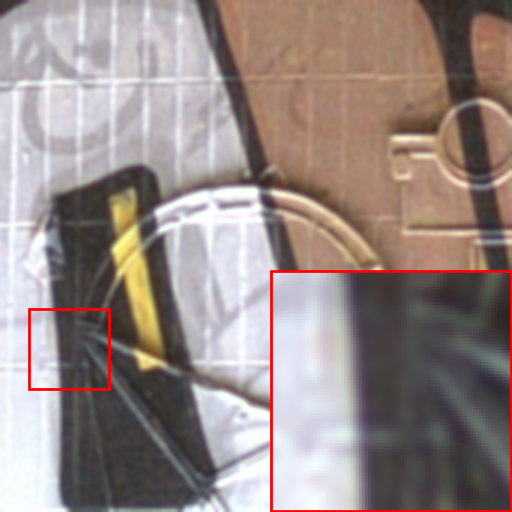}}
	\hfill
	\subfloat[AP-BSN~\cite{lee2022ap}\\
		30.79 / 0.8901  \label{fig:ap-bsn_dnd}]
	{\includegraphics[width=.19\linewidth]{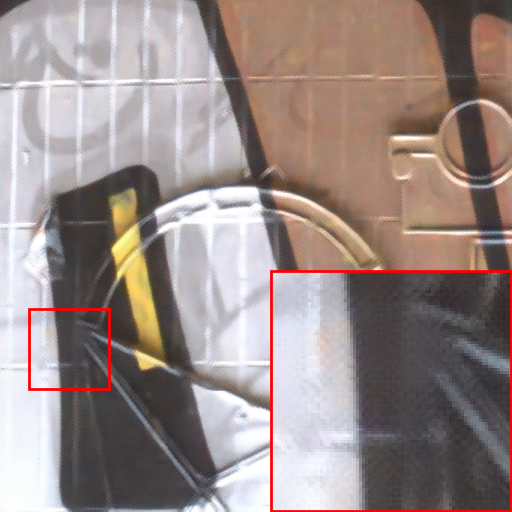}}
	\hfill
	\subfloat[AP-BSN+R$^3$~\cite{lee2022ap}\\
		32.04 / 0.9164\label{fig:ap-bsnR3_dnd}]
	{\includegraphics[width=.19\linewidth]{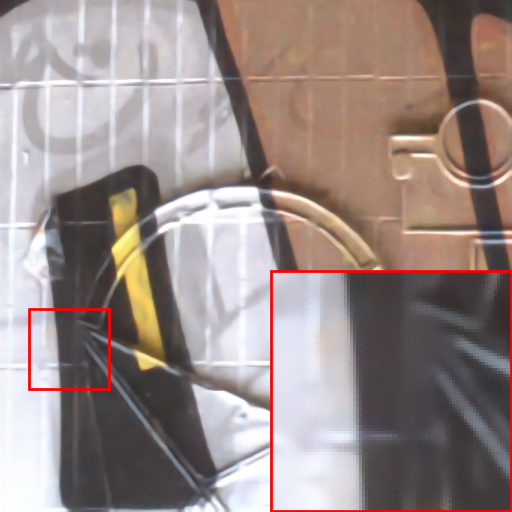}}
	\hfill
	\subfloat[C-BSN$^\dagger$\\
		\textbf{32.63 / 9180} \label{fig: ours_dnd}]
	{\includegraphics[width=.19\linewidth]{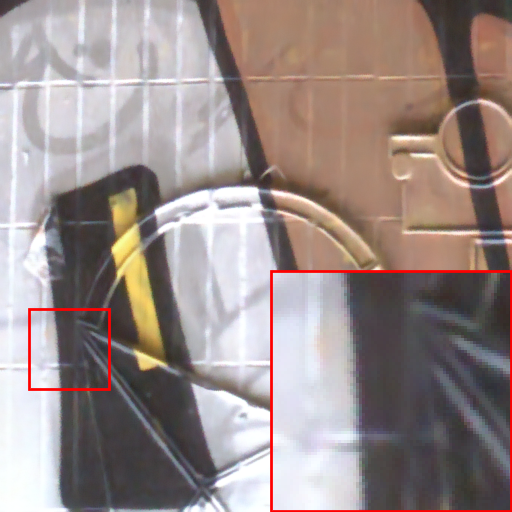}}
	\caption{\textbf{Visual comparison on DND benchmark.} PSNR and SSIM of each image are reported below.}
	\label{fig:qual}
\end{figure*}

\begin{figure*}[t]
	\captionsetup[subfigure]{justification=centering}
	\subfloat[Noisy image
		\label{fig:noisy_sidd}]
	{\includegraphics[width=.19\linewidth]{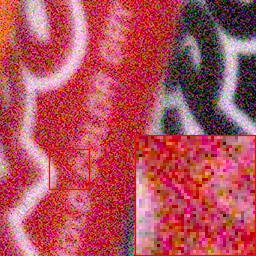}}
	\hfill
	\subfloat[CVF-SID(S$^2$)~\cite{neshatavar2022cvf}
		\label{fig:cvf_sidd}]
	{\includegraphics[width=.19\linewidth]{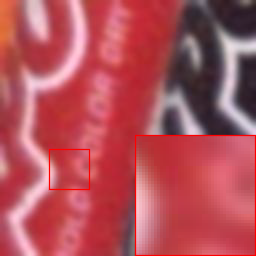}}
	\hfill
	\subfloat[AP-BSN~\cite{lee2022ap}
		\label{fig:ap-bsn_sidd}]
	{\includegraphics[width=.19\linewidth]{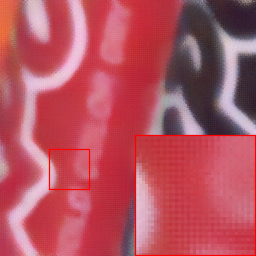}}
	\hfill
	\subfloat[AP-BSN+R$^3$~\cite{lee2022ap}
		\label{fig:ap-bsR$3_sidd}]
	{\includegraphics[width=.19\linewidth]{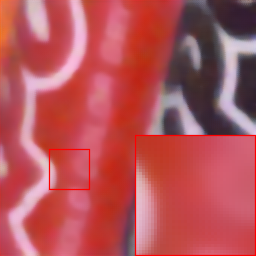}}
	\hfill
	\subfloat[C-BSN$^\dagger$\label{fig:ours_sidd}]
	{\includegraphics[width=.19\linewidth]{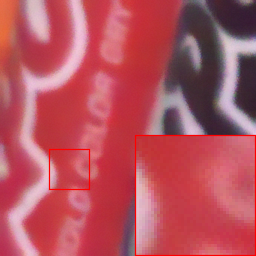}}
	\caption{\textbf{Visual comparison on SIDD benchmark.} In SIDD benchmark, PSNR and SSIM of the image is not available.}
	\label{fig:qual_sidd}
\end{figure*}

\subsection{Comparison with state-of-the-art algorithms}\label{subsec:comparison}
We compare our C-BSN against supervised, generation-based, and self-supervised methods.
The supervised models are trained on real noisy-clean pairs of SIDD, and the generation-based models simulate realistic noise and train denoiser with generated pairs.
The self-supervised models use only noisy images to train the networks.
We only report the self-supervised models that aim to remove the real noise.   
\Cref{table:Result} compares PSNR and SSIM on SIDD and DND benchmarks.
The proposed C-BSN outperforms other self-supervised methods by large margins and even some supervised networks.
C-BSN$^\dagger$ trained with the test dataset shows slightly higher PSNR than the C-BSN trained on the external dataset.
It demonstrates that the training with the same noise distribution of the test set benefits the performance of the network.
Specifically, C-BSN$^\dagger$ outperforms CVF-SID (S$^2$) and AP-BSN+R$^3$ by 2.13dB and 0.51dB, which shows the effectiveness of our framework.
The proposed downsampled invariance loss and C-BSN structure enjoy the use of blind-spot information and single inference with full image resolution. 

Figs.~\ref{fig:intro}, \ref{fig:qual}, and \ref{fig:qual_sidd} illustrate the qualitative comparisons of self-supervised methods on the DND and SIDD benchmarks.
We can see that the outputs of CVF-SID remain noisy and show stains in the flat region.
AP-BSN suffers from checkerboard artifact and AP-BSN+R$^3$ over-blur image details.
On the contrary, it can be seen that our C-BSN successfully reduces the noise and preserves the structure of the images.

Note that AP-BSN+R$^3$~\cite{lee2022ap} and CVF-SID(S$^2$)~\cite{neshatavar2022cvf} exploit a refinement technique that requires multiple runs of the network. 
AP-BSN+R$^3$ randomly replaces denoised pixels with noisy ones and averages the denoised results of randomly replaced inputs.
CVF-SID(S$^2$) trains the second model with the denoised images as a new training set and double-denoise with two successive models.
On the other hand, we do not need any post-processing and achieve state-of-the-art results with a single inference.

\subsection{Ablation Study}\label{subsec:abl}
In this section, we conduct ablation studies on the loss function, downsampler, and blind loss to show the effectiveness of the proposed method.
To reduce the cost of training, we train the networks with the patch size of $120 \times 120$ and evaluate them on the SIDD validation set.

\begin{table}
	\centering
	\caption{\textbf{Ablation on loss function.} Details of the settings of the experiment are reported in \Cref{subsec:abl}.}
	
	\begin{tabular}{>{\centering}p{0.20\textwidth}|>{\centering}p{0.08\textwidth}>{\centering\arraybackslash}p{0.08\textwidth}}
		\hline
		Loss function &PSNR(dB)&SSIM\\
		\hline
		$\mathcal{L}_{N2Same}$ &25.58&0.807\\
		$\mathcal{L}_{total}$ with blind-spot&35.86&0.931\\
		$\mathcal{L}_{inv}$ with RMS &35.63&0.920\\
		$\mathcal{L}_{total}$ & \textbf{36.22}& \textbf{0.935}\\
		\hline
	\end{tabular}
	
	\label{tab:abl_loss}
\end{table}

\noindent\textbf{Ablation on loss function.}
We analyze the different loss functions to evaluate the effectiveness of our downsampled invariance loss and conditional blind-spot network.
\Cref{tab:abl_loss} reports the PSNR on the SIDD validation dataset with four different loss functions.
For $\mathcal{L}_{N2Same}$, we set all condition $c$ to False so that the network is not blind, and the blindness is caused by masking input as in \cite{xie2020noise2same}.
The network trained with $\mathcal{L}_{N2Same}$ fails to converge, showing that a downsampling operation is necessary to reduce the spatial correlation of real noisy input.
$\mathcal{L}_{total}$ with blind-spot is trained with original D-BSN, which is not able to remove the blind-spot.
We set all $c$ to True to make the network blind while keeping the other loss functions the same.
Note that it differs from AP-BSN or D-BSN since the network is trained by $\mathcal{L}_{total}$ on full image resolution.
We can see that PSNR drops largely without C-BSN structure, which validates the importance of the center pixel information.
Lastly, $\mathcal{L}_{inv}$ with $L_2$ is trained by $\mathcal{L}_{invRS}$ with the RMS as in Noise2Same.
The performance decreases when the $L_1$ norm of $\mathcal{L}_{invRS}$ is replaced by RMS, which shows L1 norm can enhance the quality of output significantly.

\begin{figure}[t]
	\centering
	\subfloat[$PD_5$\label{fig:abl_PD5}]
	{\includegraphics[width=.47\linewidth]{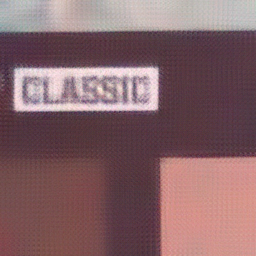}}
	\hspace{4pt}
	\subfloat[$PD_2$\label{fig:abl_PD2}]
	{\includegraphics[width=.47\linewidth]{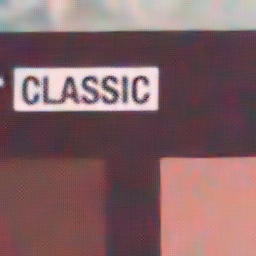}}
	\\
	\subfloat[$S2B_5$\label{fig:abl_S2B5}]
	{\includegraphics[width=.47\linewidth]{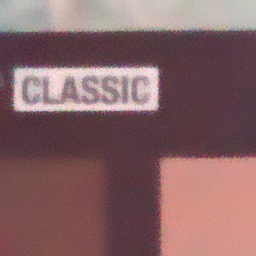}}
	\hspace{4pt}
	\subfloat[$S2B_2$\label{fig:abl_S2B2}]
	{\includegraphics[width=.47\linewidth]{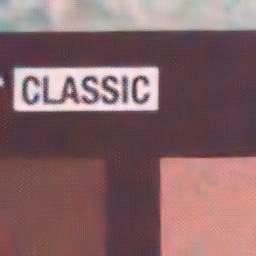}}
	\\
	\subfloat[$RS_5$\label{fig:abl_RS5}]
	{\includegraphics[width=.47\linewidth]{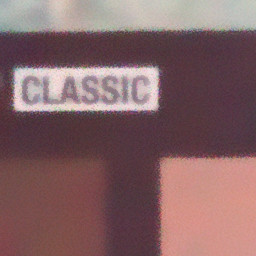}}
	\hspace{4pt}
	\subfloat[$RS_2$\label{fig:abl_RS2}]
	{\includegraphics[width=.47\linewidth]{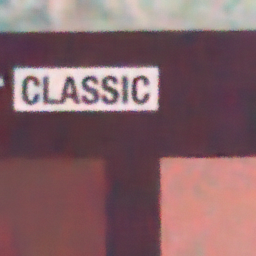}}	
	
	\caption{\textbf{Qualitative comparison of different downsampling operations in downsampled invariance loss on SIDD validation.}}
	\label{fig:down_abl}
\end{figure}

\noindent\textbf{Ablation on downsampler.}
We evaluate the networks trained with different downsamplers in the downsampled invariance loss to validate the effectiveness of our random subsampler with a stride of two. 
We test three downsamplers, PD, S2B, and RS, with strides of 2 and 5.
Each stride represents the small stride for more information and the large stride for spatial independence of real noise.
\Cref{tab:abl_ds} and \Cref{fig:down_abl} show the effectiveness of each downsample operation quantitatively and qualitatively.
As argued in \Cref{subsec:downsampling}, the networks trained with PD underperform S2B and produce visual artifacts of size $s \times s$. 
The models trained with the stride of 5 produce blurry images and cannot remove noise around the edges.
It demonstrates that it is advantageous to keep spatial information of the input with a small stride in the downsampled invariance loss.
Regardless of the stride, S2B outperforms PD, and RS outperforms S2B.
PD and S2B with a stride of two can reduce spatially correlated noise, but it also produces severe checkerboard artifacts.
On the other hand, the proposed $RS_2$ achieves the highest PSNR and visually pleasing result without artifacts, outperforming $PD_2$ and $S2B_2$ by 0.90dB and 0.20dB, respectively.

\noindent\textbf{Ablation on the blind loss.}
We investigate the effectiveness of the blind loss, $\mathcal{L}_{blind}$.
Though $\mathcal{L}_{CBSN}$ is an upper bound of the supervised loss, the training is unstable without $\mathcal{L}_{blind}$.
We set the hyperparameter $\lambda_{sch}$ to different conditions as in \Cref{tab:abl_blind}.
When $\lambda_{sch}=\infty$, we do not use the blind loss and train C-BSN with $\mathcal{L}_{CBSN}$ only.
In this case, the network fails to learn denoising and outputs zeros, resulting in a flat image of the input mean.
With $\lambda_{sch}=0$, the loss function is $\mathcal{L}_{blind}$ as AP-BSN~\cite{lee2022ap}. However, processing AP-BSN with the original size input without a blind-spot produces severe artifacts and poor image quality.
It can be seen that $\lambda_{sch}=1$ shows suboptimal PSNR to warm-up, yet it sometimes falls to the same local optima as $\lambda_{sch}=\infty$.
The suggested warm-up scheduling brings about 0.57dB PSNR improvement and stabilizes the training procedure.

\section{Conclusion}
We have presented a novel self-supervised image denoising framework C-BSN for real camera noise reduction. We have derived the downsampled invariance loss, which is the upper bound of the supervised loss and enables the training without a blind-spot.
The C-BSN structure conditionally controls blind-spot, and then the random subsampler decorrelates noise without introducing visual artifacts.
Without using post-processing or refinement, our C-BSN outperforms recent self-supervised denoisers.
\begin{table}[t]
	\centering
	\caption{\textbf{Ablation on the downsampler of downsampled invariance loss.}}
	\begin{tabular}{>{\centering}p{0.1\textwidth}|>{\centering}p{0.05\textwidth}|>{\centering}p{0.1\textwidth}>{\centering\arraybackslash}p{0.1\textwidth}}
		\hline
		downsampler & stride & PSNR(dB) & SSIM \\ \hline
		\multirow{2}{*}{$PD$} &5& 34.71 & 0.905 \\
		&2& 35.32 & 0.914 \\ \hline
		\multirow{2}{*}{$S2B$} &5& 35.62 & 0.924 \\
		&2& 36.02 & 0.922 \\ \hline
		\multirow{2}{*}{$RS$} &5& 35.24 &  0.922 \\
		&2& \textbf{36.22} & \textbf{0.935} \\
		\hline
	\end{tabular}
	\label{tab:abl_ds}
\end{table}

\begin{table}[t]
	\centering
	\caption{\textbf{Ablation on the blind loss.}}
	\begin{tabular}{>{\centering}p{0.15\textwidth}|>{\centering}p{0.1\textwidth}>{\centering\arraybackslash}p{0.1\textwidth}}
		\hline
		$\lambda_{sch}$ & PSNR(dB) & SSIM \\ \hline
		$\infty$ & 25.92 & 0.810 \\
		0 & 29.59& 0.757\\
		1 & 35.65 & 0.926 \\
		warm-up & \textbf{36.22} & \textbf{0.935} \\
		\hline
	\end{tabular}
	\label{tab:abl_blind}
\end{table}

\section*{Appendix}
\renewcommand{\thepage}{S\arabic{page}}
\renewcommand{\thesection}{S\arabic{section}}
\renewcommand{\thetable}{S\arabic{table}}
\renewcommand{\thefigure}{S\arabic{figure}}
\renewcommand{\theequation}{s\arabic{equation}}

\setcounter{table}{0}
\setcounter{figure}{0}
\setcounter{equation}{0}
\setcounter{section}{0}
\setcounter{proposition}{0}

\section{Detailed Proof of Downsampled Invariance Loss}
\label{supp:proof}
\begin{proposition}\label{supp:proposition}
	Let $\boldsymbol{x}$ be a normalized zero-mean noisy image conditioned on $\boldsymbol{y}$, $\E[\boldsymbol{x}|\boldsymbol{y}]=\boldsymbol{y}$. 
	Let $d$ be any downsampling operation and $d_s(\boldsymbol{x})$ be a set of downsampled pixels of $\boldsymbol{x}$ with a stride of $s$.
	Assume that downsampled subimage $d_s(\boldsymbol{x})$ has zero pixel-wise correlation and $f_M$ is a blind-spot network. Then, the following inequality holds.
	\begin{align}
		&\E_{\boldsymbol{x},\boldsymbol{y}} \norm{f(\boldsymbol{x})-\boldsymbol{y}}^2 + \norm{\boldsymbol{x}-\boldsymbol{y}}^2 \leq 
		\E_{\boldsymbol{x}} \norm{f(\boldsymbol{x})-\boldsymbol{x}}^2\nonumber &&\\
		&\hspace{6mm}+ 2\sqrt{ms^2} \mathop{\E}_{d_s(\boldsymbol{x})}[\E \,\norm{d_s(f(\boldsymbol{x})) - f_{M}(d_s(\boldsymbol{x}))}^2]^{\frac{1}{2}}.
		\label{supp:eq:prop}
	\end{align}
\end{proposition}

\begin{proof}
	We follow similar steps with the supplementary material of \cite{xie2020noise2same}. 
	Self-supervised loss can be decomposed as 
	\begin{align}
		\E_{\boldsymbol{x}} ||f(\boldsymbol{x})-\boldsymbol{x}||^2  = &\E_{\boldsymbol{x},\boldsymbol{y}} ||f(\boldsymbol{x})-\boldsymbol{y}||^2 + ||\boldsymbol{x}-\boldsymbol{y}||^2\nonumber \\
		&-2 \langle f(\boldsymbol{x})-\boldsymbol{y},\boldsymbol{x}-\boldsymbol{y} \rangle.
	\end{align}
	Then, \Cref{supp:proposition} is equivalent to that the third term $\langle f(\boldsymbol{x})-\boldsymbol{y},\boldsymbol{x}-\boldsymbol{y} \rangle$ is upper-bounded by the rightmost term in \cref{supp:eq:prop}.
	$\langle f(\boldsymbol{x})-\boldsymbol{y},\boldsymbol{x}-\boldsymbol{y} \rangle$ can be formulated as 
	\begin{align}
		&\E_{\boldsymbol{x},\boldsymbol{y}} \langle f(\boldsymbol{x})-\boldsymbol{y},\boldsymbol{x}-\boldsymbol{y} \rangle \\
		&= \E_{\boldsymbol{y}} \E_{\boldsymbol{x}|\boldsymbol{y}} \sum_{j} (f(\boldsymbol{x})_j-y_j)(x_j-y_j)\\
		&= \sum_{j} \E_{\boldsymbol{y}} [\E_{\boldsymbol{x}|\boldsymbol{y}}(f(\boldsymbol{x})_j-y_j)(x_j-y_j) \nonumber\\
		&-\E_{\boldsymbol{x}|\boldsymbol{y}}(f(\boldsymbol{x})_j-y_j)\E_{\boldsymbol{x}|\boldsymbol{y}}(x_j-y_j)]\label{eq:1_3}\\ 
		&= \sum_{j} \E_{\boldsymbol{y}} [\mathrm{Cov}(f(\boldsymbol{x})_j-y_j,x_j-y_j|\boldsymbol{y})]\\
		&= \sum_{j} \E_{\boldsymbol{y}} [\mathrm{Cov}(f(\boldsymbol{x})_j,x_j|\boldsymbol{y})].
		\label{supp:eq:1}
	\end{align}
	\cref{eq:1_3} holds since $\E_{x|y}(x_j-y_j)=0$ by the zero-mean noise assumption.
	Let $J$ be a subset of the image sampled by a random downsampling operation $d_s(\boldsymbol{x})$. Then we have the equation,
	\begin{equation}
		\sum_{j} \E_{\boldsymbol{y}} [\mathrm{Cov}(f(\boldsymbol{x})_j,x_j|\boldsymbol{y})] = \frac{m}{|J|} \E_{J} \sum_{j} \E_{\boldsymbol{y}} [\mathrm{Cov}(f(\boldsymbol{x})_j,x_j|\boldsymbol{y})],
		\label{supp:eq:2}
	\end{equation}
	\noindent since every pixel has the chance of selecting $|J|/m=1/s^2$.
	On the right-hand side, the covariance term can be upper-bounded as
	
	\begin{align}
		& \frac{1}{|J|} \sum_{j \in J}\E_{\boldsymbol{y}} [\mathrm{Cov}(f(\boldsymbol{x})_j,x_j|\boldsymbol{y})] \\
		&= \frac{1}{|J|} \sum_{j \in J}\E_{\boldsymbol{y}} [\mathrm{Cov}(f(\boldsymbol{x})_j - f_{M}(d_s(\boldsymbol{x}))_j,x_j|\boldsymbol{y})]\label{supp:eq:3}\\
		& \leq \frac{1}{|J|} \sum_{j \in J}(\E_{\boldsymbol{y}} [\mathrm{Var}(f(\boldsymbol{x})_j - f_{M}(d_s(\boldsymbol{x}))_j|\boldsymbol{y})^{\frac{1}{2}} \cdot \mathrm{Var}(x_j|\boldsymbol{y})^{\frac{1}{2}}])\label{supp:eq:4}\\
		& \leq (\frac{1}{|J|} \sum_{j \in J}\E_{\boldsymbol{y}} [\mathrm{Var}(f(\boldsymbol{x})_j - f_{M}(d_s(\boldsymbol{x}))_j|\boldsymbol{y}) \cdot \mathrm{Var}(x_j|\boldsymbol{y})])^{\frac{1}{2}}\label{supp:eq:4_2}\\
		& \leq (\frac{1}{|J|} \sum_{j \in J}\E_{\boldsymbol{y}} [E[(f(\boldsymbol{x})_j - f_{M}(d_s(\boldsymbol{x}))_j)^2|\boldsymbol{y}]] \cdot 1)^{\frac{1}{2}}]\label{supp:eq:5}\\
		& = (\frac{1}{|J|} \sum_{j \in J}\E[(f(\boldsymbol{x})_j - f_{M}(d_s(\boldsymbol{x}))_j)^2])^{\frac{1}{2}}\label{supp:eq:5_2}\\
		& = (\frac{s^2}{m} \E[(d_s(f(\boldsymbol{x})) - f_{M}(d_s(\boldsymbol{x})))^2])^{\frac{1}{2}}\label{supp:eq:6}
	\end{align}
	\noindent	 In \cref{supp:eq:3}, the equality holds since $x_j$ is excluded in BSN and downsampled surroundings have no correlation with $x_j$ by the assumption. Note that the Inequality (\ref{supp:eq:4}) is derived from the Cauchy-Schwarz inequality, and the Inequality (\ref{supp:eq:4_2}) is derived from Jensen's inequality.
	Also, the Inequality (\ref{supp:eq:5}) holds by the fact that $ \mathrm{Var}(x) \leq E[x^2]$, and by the assumption that input $\boldsymbol{x}$ is normalized \ie,
	$\mathrm{Var}(x_j|\boldsymbol{y}) \leq \mathrm{Var}(x_j) = 1.$
\end{proof}

By the \Cref{proposition}, we use \cref{supp:eq:6} as downsampled invariance loss, 
\begin{equation}
	\mathcal{L}_{inv} = \sqrt{\dfrac{s^2}{m}}||d_s(f(\boldsymbol{x})) -sg(f_M(d_s(\boldsymbol{x})))||_2,
	\label{supp:eq:inv}
\end{equation}
\noindent 
where $sg$ is the stop gradient operation. 
$f_M(d_s(\boldsymbol{x}))$ is introduced to \cref{supp:eq:3} since it has zero correlation with $x_j$. Therefore, we regard it as a constant and adopt a stop-gradient operation in the loss function.
Lastly, we replace the root mean squared error with mean absolute difference in downsampled invariance loss as 
\begin{equation}
	\mathcal{L}_{inv} = \dfrac{s^2}{m}||d_s(f(\boldsymbol{x})) -sg(f_M(d_s(\boldsymbol{x})))||_1.
	\label{supp:eq:inv_final}
\end{equation}

\section{Analysis of Downsampling Ratio in Loss Functions}
We conduct extensive experiments to analyze the effects of the downsampling ratios in $\mathcal{L}_{invRS}$ and $\mathcal{L}_{blind}$.
\Cref{supp:fig:abl_down} shows the PSNR of C-BSN$_{a/b}$ on SIDD validation dataset~\cite{abdelhamed2018high}, where $a$ is the stride of RS in the downsampled invariance loss and $b$ is the stride of S2B in the blind loss.

Using strides less than 4 in the blind loss leads to suboptimal performance, showing that reducing spatial correlation of masked network input is crucial.
Regarding the strides of RS, the performance tends to decrease as the stride increases over 3, while C-BSN with $a=1$ fails to denoise the image.
Although the performance is maximized with C-BSN$_{3/4}$, the performance gap is marginal and falls within the range of variation caused by the randomness of the training process.
Therefore, we adopt C-BSN$_{2/5}$ as a baseline, consistent with AP-BSN~\cite{lee2022ap}.

\section{Ablation on Downsampler of Blind Loss}
We conduct an additional ablation study on the downsampler of blind loss.
We follow the same setting as Section 4.3 in the paper. 
\Cref{supp:tab:abl_ds} reports PSNR and SSIM of the network with different downsampler in the blind loss. 
Regardless of downsampling operations, models trained with small strides show poor performance, which is consistent with the result of \Cref{supp:fig:abl_down}.
Space2batch, with a stride of 5, achieves the highest PSNR and SSIM compared to the other two downsamplers.
Therefore, we employ S2B as the downsampling function for the blind loss.

\begin{figure}[t]
	\centering
	\includegraphics[width=\linewidth]{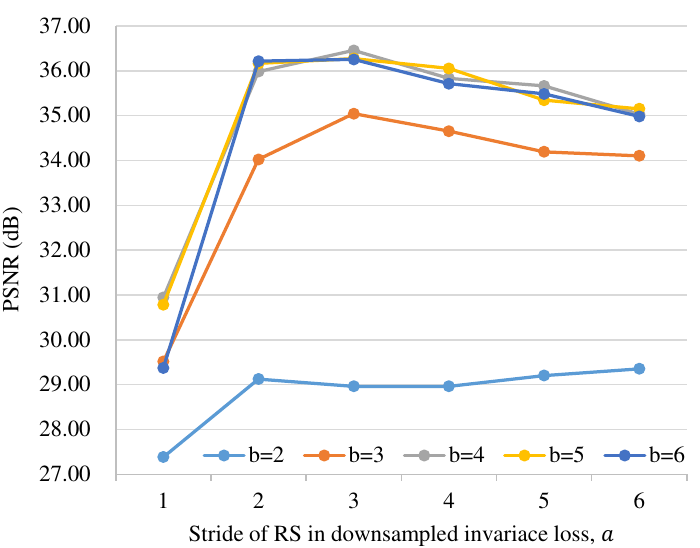}
	\caption{\textbf{PSNR of C-BSN$_{a/b}$ on SIDD validation~\cite{abdelhamed2018high},}
		where $a$ denotes the stride of RS and $b$ denotes the stride of S2B.
	}
	\label{supp:fig:abl_down}
\end{figure}

\section{More Visualized Results}
We present more visual comparisons on SIDD~\cite{abdelhamed2018high} validation and NIND~\cite{brummer2019natural}.
We compare C-BSN with other self-supervised methods, CVF-SID (T)~\cite{neshatavar2022cvf}, CVF-SID (S$^2$), AP-BSN~\cite{lee2022ap}, AP-BSN (R$^3$)~\cite{lee2022ap}, which aim to remove real-world noise.
We use official code from the authors' GitHub with the pre-trained model. 
The denoised results of various scenes are illustrated in \Cref{supp:fig:SIDDval}.

For NIND, we use C-BSN$^\dagger$ which is trained on the test set directly.
\Cref{supp:fig:NIND} shows the noisy images from NIND and its denoised outputs.
We mark ROI with red boxes for each image and present noisy-denoised pairs of cropped patches.

\begin{table}[t]
	\centering
	\caption{\textbf{Ablation on the downsampler of blind Loss.}}
	\begin{tabular}{>{\centering}p{0.1\textwidth}|>{\centering}p{0.05\textwidth}|>{\centering}p{0.1\textwidth}>{\centering\arraybackslash}p{0.1\textwidth}}
		\hline
		downsampler & stride & PSNR(dB) & SSIM \\ \hline
		\multirow{2}{*}{$PD$} &5& 34.83 & 0.912 \\
		&2& 29.11 & 0.715 \\ \hline
		\multirow{2}{*}{$S2B$} &5& \textbf{36.22} & \textbf{0.935} \\
		&2& 25.93 & 0.810 \\ \hline
		\multirow{2}{*}{$RS$} &5& 35.67 &  0.924 \\
		&2& 30.54 & 0.771 \\
		\hline
	\end{tabular}
	\label{supp:tab:abl_ds}
\end{table}

\begin{figure*}[t]
	\centering
	\captionsetup[subfigure]{position=bottom}
	{\includegraphics[width=.98\linewidth]{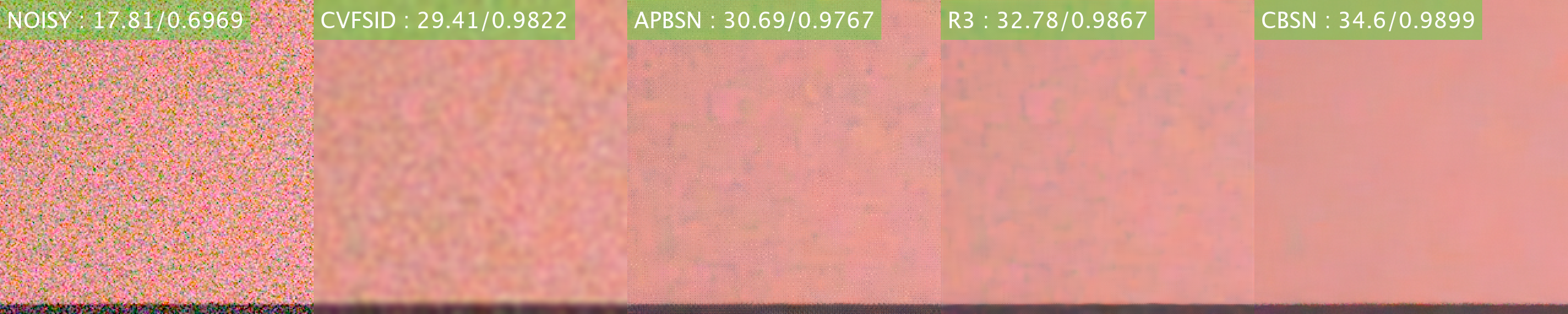}}
	{\includegraphics[width=.98\linewidth]{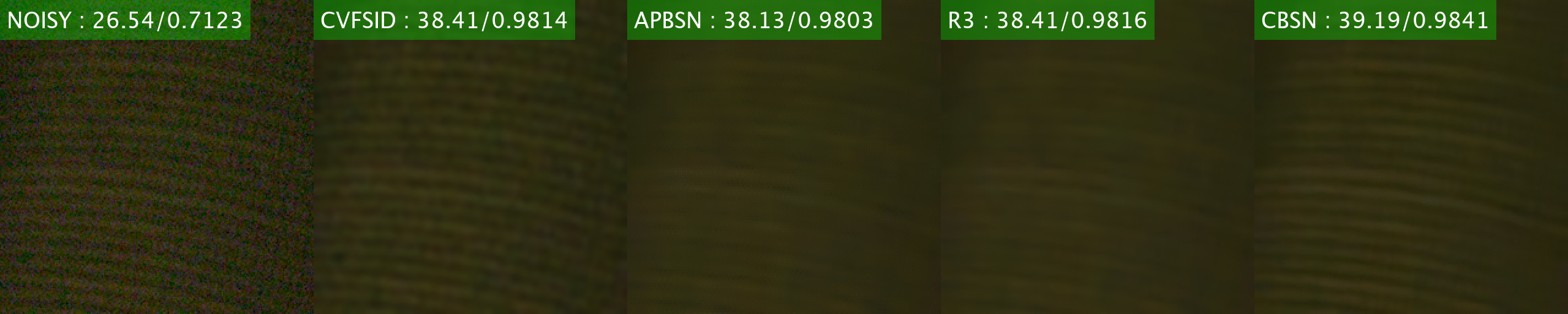}}
	{\includegraphics[width=.98\linewidth]{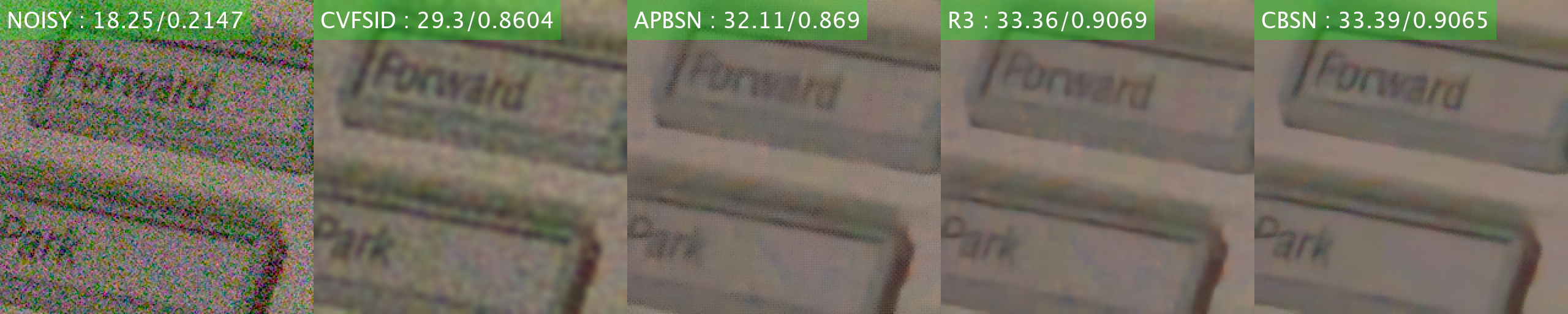}}
	{\includegraphics[width=.98\linewidth]{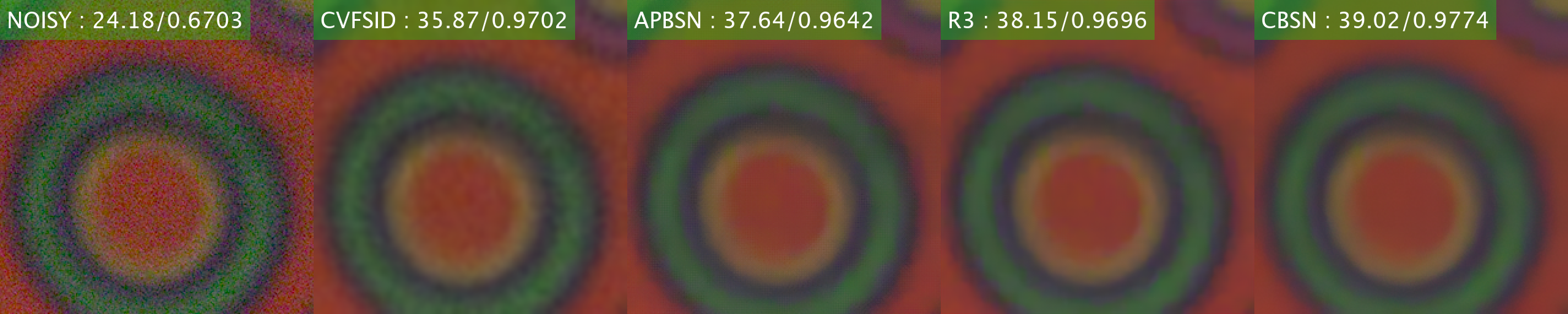}}
	{\includegraphics[width=.98\linewidth]{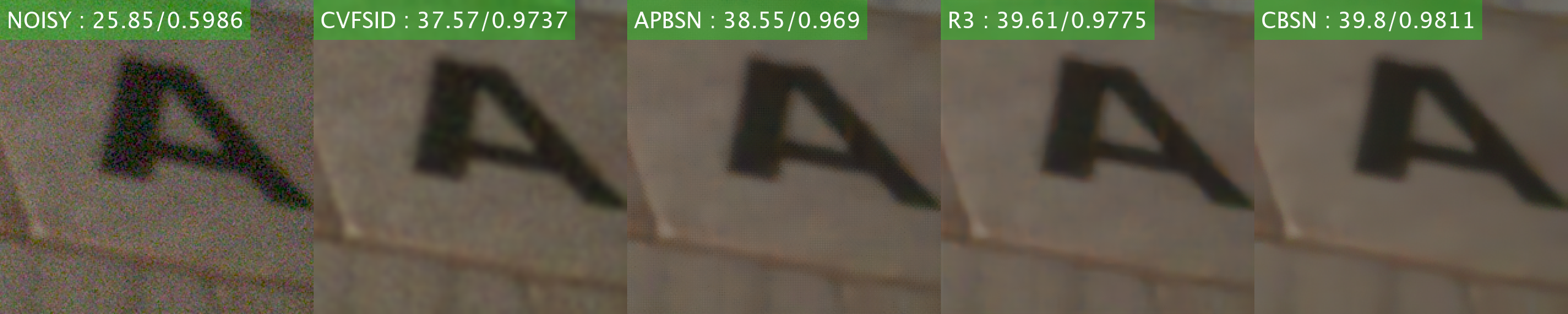}}
	{\includegraphics[width=.98\linewidth]{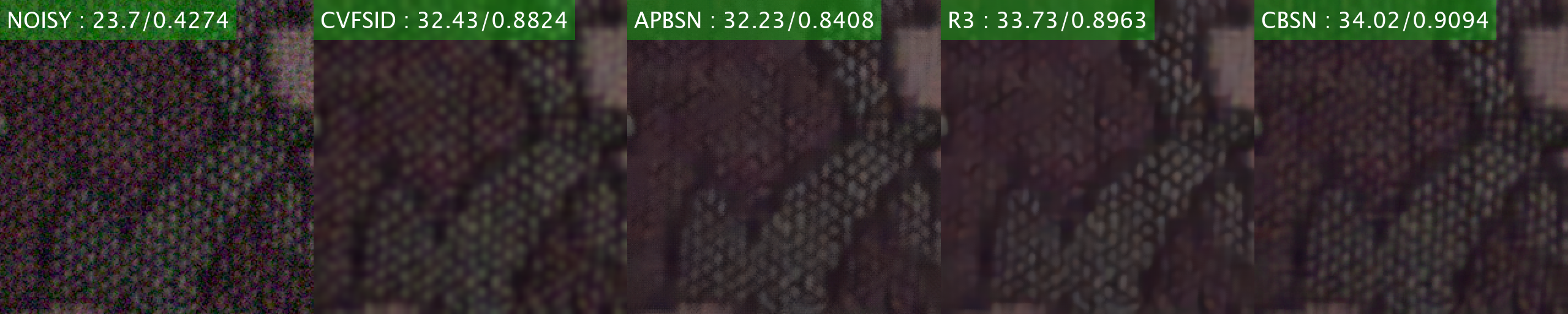}}
	\caption{\textbf{Visual comparison of denoised images on SIDD validation~\cite{abdelhamed2018high}.} 
		We provide PSNR and SSIM in the upper left of the images.
		All images are upsampled by 2 with the nearest neighbor for better comparison.
		Best viewed in pdf.
	}
	\label{supp:fig:SIDDval}
\end{figure*}

\begin{figure*}[t]
	\centering
	\captionsetup[subfigure]{position=bottom}
	\subfloat[NIND\_soap\_ISO6400]{
		\begin{minipage}{.64\linewidth}
				\includegraphics[width=\textwidth]{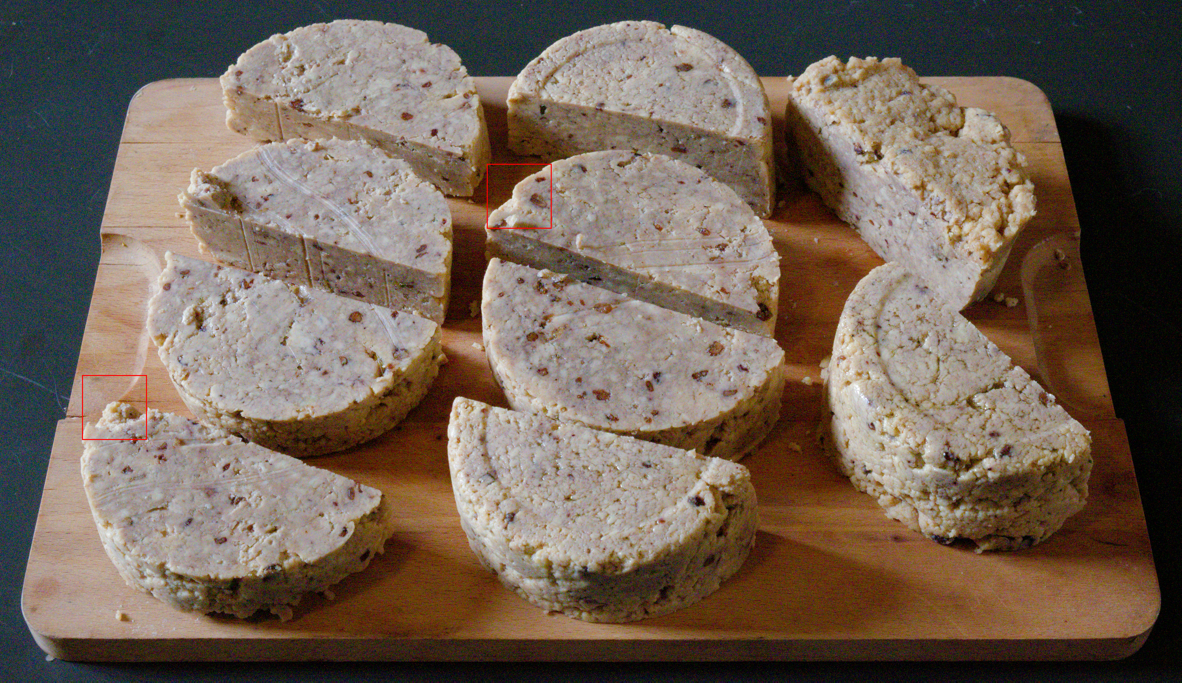}
		\end{minipage}
		\begin{minipage}{.36\linewidth}
				\includegraphics[width=\textwidth]{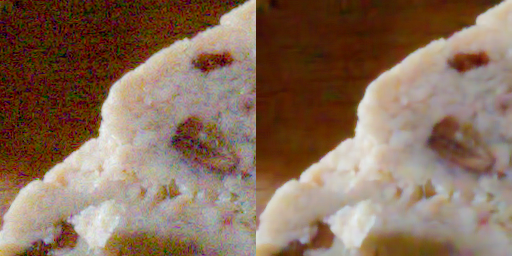}\vspace{4pt}\\
				\includegraphics[width=\textwidth]{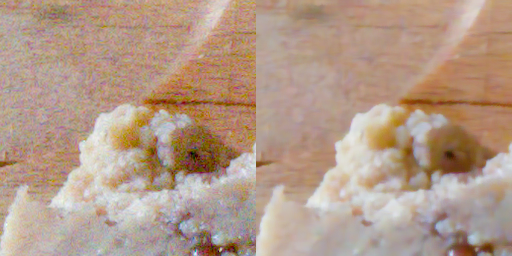}
	\end{minipage}}\\
	\subfloat[NIND\_MuseeL-coral2\_ISOH1]{
		\begin{minipage}{.64\linewidth}
				\includegraphics[width=\textwidth]{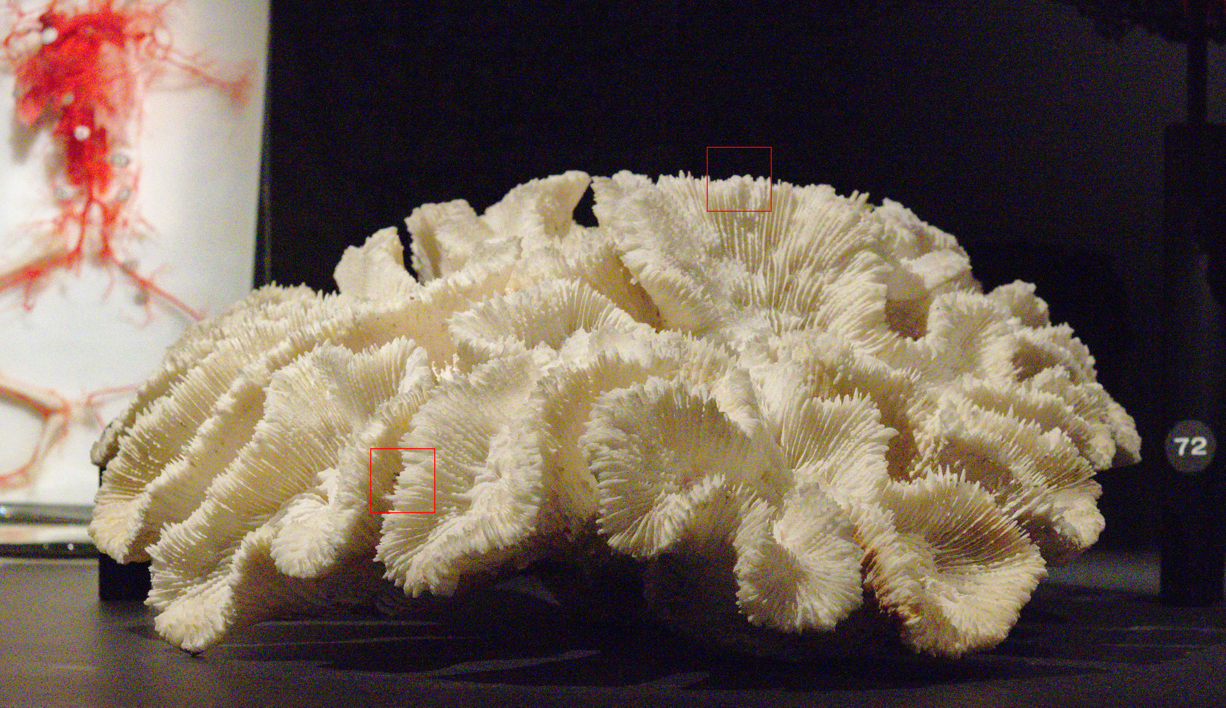}
		\end{minipage}
		\begin{minipage}{.36\linewidth}
				\includegraphics[width=\textwidth]{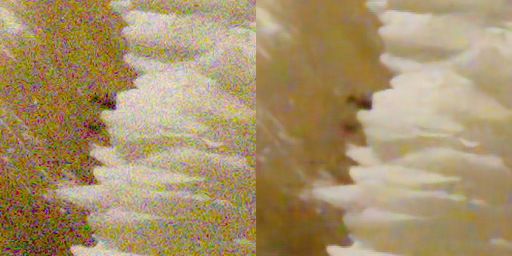}\vspace{4pt}\\
				\includegraphics[width=\textwidth]{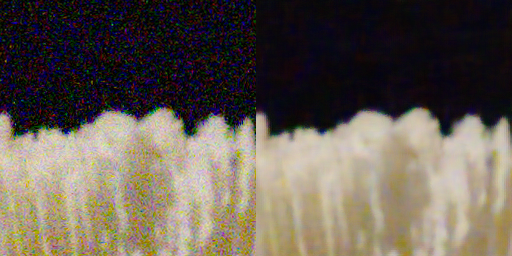}
	\end{minipage}}\\
	\subfloat[NIND\_MVB-LouveFire\_ISOH1]{
		\begin{minipage}{.64\linewidth}
				\includegraphics[width=\textwidth]{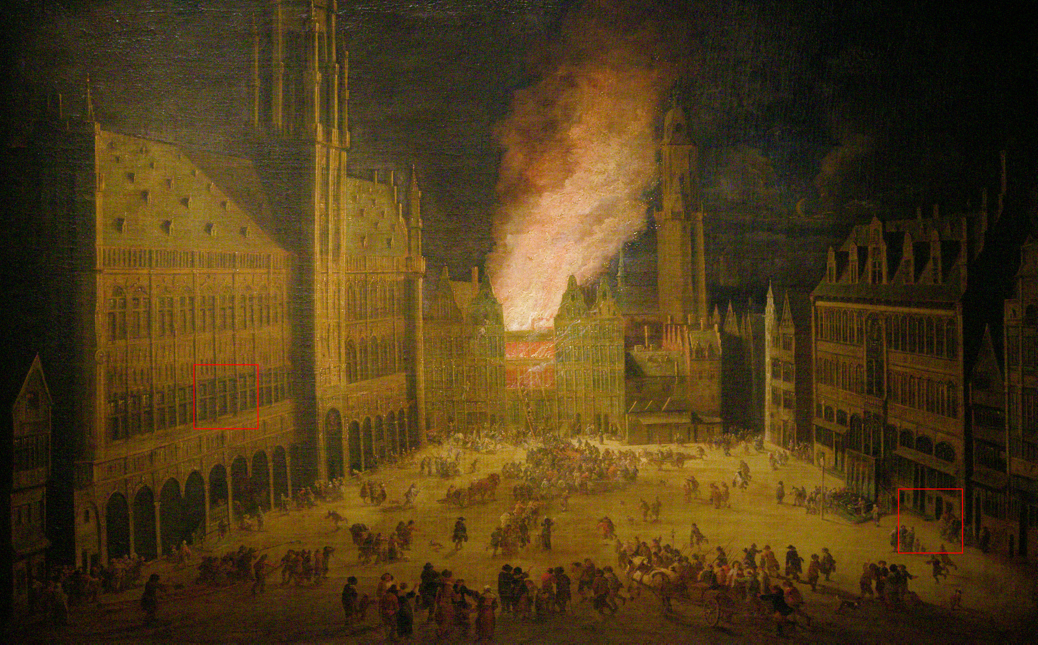}
		\end{minipage}
		\begin{minipage}{.36\linewidth}
				\includegraphics[width=\textwidth]{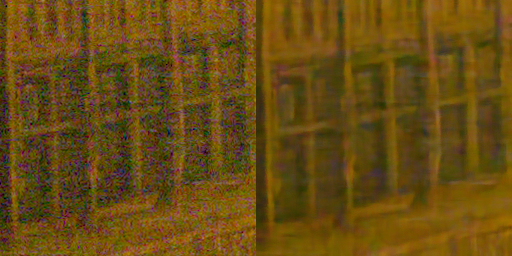}\vspace{4pt} \\
				\includegraphics[width=\textwidth]{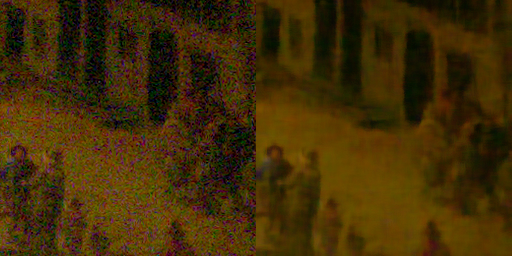}
	\end{minipage}}
	
	\caption{\textbf{C-BSN$^\dagger$ results of  NIND~\cite{brummer2019natural} samples.} (Left) Real noisy images from NIND.
		(Right) Enlarged noisy-Denoised image pairs.}
	\label{supp:fig:NIND}
\end{figure*}

\clearpage
%%%%%%%%% REFERENCES
{\small
	\bibliographystyle{ieee_fullname}
	\bibliography{arxiv_version}
}

\end{document}